\documentclass[runningheads]{llncs}
\usepackage[T1]{fontenc}
\usepackage{graphicx}
\usepackage{amsmath,amssymb,amsfonts}
\usepackage{bm}
\usepackage[title]{appendix}%
\usepackage{xcolor}%
\usepackage{textcomp}%
\usepackage{manyfoot}%
\usepackage{booktabs}%
\usepackage{algorithm}%
\usepackage{algorithmicx}%
\usepackage{algpseudocode}%
\usepackage{listings}%
\usepackage{bbding}
\begin{document}
\title{Constrained Distributed Heterogeneous Two-Facility Location Problems with Max-Variant Cost}
\titlerunning{Constrained Distributed Heterogeneous Two-Facility Location Problems}
%
\author{Xinru Xu\inst{1} \and
Wenjing Liu\Envelope\inst{1,2} \and
Qizhi Fang\inst{1}}
\authorrunning{X. Xu et al.}
\institute{ School of Mathematical Sciences, Ocean University of China, Qingdao, 266100, China \and Laboratory of Marine Mathematics, Ocean University of China, Qingdao, 266100, China\\
\email{xuxinru1207@stu.ouc.edu.cn},
\email{\{liuwj,qfang\}@ouc.edu.cn}}
%

%
%
\maketitle              
\begin{abstract}
We study a constrained distributed heterogeneous two-facility
location problem, where a set of agents with private locations on the real line  are divided into disjoint groups. The constraint means that the facilities can only be built in a given multiset of candidate locations and at most one facility can be built at each candidate location.  Given the locations of the two facilities, the cost of an agent is the distance from her location to the farthest facility (referred to as max-variant). Our goal is to design strategyproof distributed mechanisms that can incentivize all agents to truthfully report their locations and approximately optimize some social objective. A distributed mechanism consists of two steps: for each group, the mechanism chooses two candidate locations as the representatives of the group based only on the locations reported by agents therein; then, it outputs two facility locations among all the representatives. We focus on a class of deterministic strategyproof distributed mechanisms and analyze upper and lower bounds on the distortion under the Average-of-Average cost (average of the average individual cost of agents in each group), the Max-of-Max cost (maximum individual cost among all agents), the Average-of-Max cost (average of the maximum individual cost among all agents in each group) and the Max-of-Average cost (maximum of the average individual cost of all agents in each group). Under four  social objectives, we obtain constant upper and lower distortion bounds.

\keywords{Mechanism design without money  \and Facility location \and Strategyproof \and Distributed \and Distortion}
\end{abstract}
\section{Introduction}

The facility location problem is a classic combinatorial optimization problem. Its main goal is to select the optimal locations for facilities under given constraints to optimize some social objective. However, in many real-world scenarios, the information of agents (e.g., their residential addresses) may be private. Therefore, social planners can only locate facilities based on the information reported by agents. On the one hand,  social planners aim to optimize some social objective. On the other hand, selfish agents may misreport their information in order to minimize their individual  costs. Naturally, social planners want to design mechanisms  that can ensure each agent's truthful report (e.g., strategyproof) while (approximately) optimizing some social objective. Procaccia and Tennenholtz \cite{procaccia2013approximate} were the first to study  approximate mechanism design without payment for facility location problems. Subsequently, various different facility location game models have been proposed; see the survey by Chan et al. \cite{chan2021mechanism}.

 In real-world scenarios, a  collective decision-making process may be distributed, as follows: a set of agents are divided into groups, and each group makes a decision first (without considering the agents of the other groups), and then these decisions  are aggregated into a collective decision. For example, in the selection of outstanding students at a university, first each faculty  selects its outstanding student representatives, and then the university-level outstanding students are finally selected based on these representatives.  To analyze these more complex problems, Filos-Ratsikas et al. \cite{filos2020distortion} initiated the study of social choice problems in distributed settings, in which decisions are made by two-step mechanisms: for each group, the mechanism first selects a representative  based on the local election with the agents therein, and then outputs one of the representatives as the winner. To quantify the inefficiency of distributed mechanisms, Filos-Ratsikas et al. \cite{filos2020distortion}  extended the notion of distortion, which is broadly used in social choice problems. Distortion of a mechanism refers to the worst-case ratio (over all instances) between the social objective value obtained by the mechanism and the optimal social objective value. In follow-up work, Filos-Ratsikas and Voudouris \cite{filos2021approximate} studied a distributed single-facility location problem under minimizing the social cost objective (i.e., the sum of individual costs of all agents).

In the classic models,  facilities can be built at any point in a metric space. However, in many real-world scenarios,  the feasible regions for building facilities and  the number of facilities that can be built at each  location may  be limited, due to land use restrictions, environmental protection, humanity factors, etc.  Motivated by this,  mechanism design for facility location problems with limited locations has been studied, which is termed as constrained  facility location. 

 In this paper, we study a constrained distributed heterogeneous two-facility location problem. Here, we assume each agent approves the two facilities and  the individual cost of each agent is the distance from her location to the farthest facility (referred to as max-variant cost). As for max-variant, consider a scenario where  an express delivery outlet needs to transport regular packages to a ordinary distribution center  and cold chain packages to a professional distribution center  respectively. Assuming the outlet has multiple transport vehicles with the same speed, its waiting time depends on the distance to the farthest distribution center. We show upper and lower bounds on the distortion of strategyproof distributed mechanisms under  four social objectives. More details are provided below.

\subsection{Our Results}

We  study a constrained distributed heterogeneous two-facility location problem with max-variant, where  the facilities can only be built in a given set of candidate locations and at most one facility can be built  at each candidate location. A set of agents have private locations on the real line and they  are divided into disjoint groups. Once the two facilities have been located, the individual cost of each agent is the distance from her location to the farthest facility.
Our goal is to design strategyproof distributed mechanisms that take as input the locations reported by the agents and output the locations of the  two facilities. A distributed mechanism consists of two steps: for each group, the mechanism chooses two candidate locations as the representatives of the group based only on the locations reported by agents therein; then, it outputs two facility locations among all the representatives. 

Following the work of Anshelevich et al. \cite{anshelevich2022distortion}, we focus on the following four social objectives: the average of the average individual cost of all agents in each group (i.e., Average-of-Average cost); the maximum individual cost among all agents (i.e., Max-of-Max cost); the maximum of the average individual cost of all agents in each group (i.e., Max-of-Average cost); and the average of the maximum individual cost among all agents in each group (i.e., Average-of-Max cost). While the Average-of-Average cost and the Max-of-Max cost are adaptations of objectives that have been considered in the classic setting (i.e., non-distributed setting), the    Max-of-Average cost and the  Average-of-Max cost  are    fairness-inspired objectives that are only meaningful for the distributed setting. Under these  four social objectives, we show  upper and lower bounds on the distortion of  strategyproof distributed mechanisms.
  A summary of our results is shown in Table 1.
\begin{table}\centering
\caption{Upper and lower bounds on the distortion of strategyproof distributed  mechanisms.}\label{tab1}
\begin{tabular}{|l|c|c|}
\hline
Social objective &  Upper bound & Lower bound\\
\hline
Average-of-Average cost &  $9$ & $3$\\
Max-of-Max cost &  $3$ & $3$\\
Max-of-Average cost&$2+\sqrt{5}$& $\frac{7}{2}$\\
Average-of-Max cost& $2+\sqrt{5}$& $3$\\
\hline
\end{tabular}
\end{table}

\subsection{Related work}

Approximate mechanism design without money  was initiated by Procaccia and Tennenholtz \cite{procaccia2013approximate}, who studied  strategyproof mechanisms with constant approximation ratios for facility location problems on the line under the social cost objective and the maximum cost objective (i.e., the maximum individual cost among all agents). They considered single-facility location problems and homogeneous two-facility location problems  where each agent's individual cost is the distance from her location to the closest facility (referred to as min-variant cost). Since then, numerous different truthful facility location problems have been well studied. For example, Alon et al. \cite{alon2010strategyproof} studied single-facility location problems in circles and general graphs. Tang et al. \cite{tang2020mechanism} considered  constrained single and two facility location problems where the facility can only be built in a given set of candidate locations of the line. Cheng et al. \cite{cheng2013strategy} discussed an obnoxious facility location game where every agent wants to stay far away from the facility. Cai et al. \cite{cai2016facility} studied   facility location games under the minimax envy objective and Ding et al. \cite{ding2020facility}  studied the envy ratio objective. We will highlight the work of heterogeneous facility location and distributed facility location problems respectively, which is most related to our work.

\textbf{Heterogeneous facility location.}
Zou and Li \cite{zou2015facility} studied heterogeneous facility location problems with dual preferences where  the facility can be desirable or obnoxious for every agent. Serafino and Ventre \cite{serafino2016heterogeneous} considered   heterogeneous facility location problems with optional preferences where each agent approves either one facility or both.  In their setting,  the agents' locations are public while agents' preferences are private and the individual cost of each agent is the sum of distances to her interested facilities (referred to as sum-variant cost).  Later, Chen et al. \cite{chen2020facility} considered the optional preference model with max-variant and min-variant; and their results with min-variant cost  were improved by Li et al. \cite{li2021strategyproof}. Anastasiadis and
Deligkas \cite{anastasiadis2018heterogeneous} studied heterogeneous $k$-facility location problems with min-variant. When the  locations of agents are private and the preferences of agents are public, Zhao et al. \cite{zhao2023constrained} studied the optional preference model with max-variant cost; Kanellopoulos et al. \cite{kanellopoulos2025truthful} studied the optional preference model with sum-variant cost. Fong et al. \cite{fong2018facility} proposed a  fractional preference model where the preference of each agent for the facility is a number between $0$ and $1$.

\textbf{Distributed facility location.} Filos-Ratsikas et al. \cite{filos2020distortion} initiated the study of the distortion in distributed social choice and then Anshelevich et al. \cite{anshelevich2022distortion} considered a distributed metric social choice setting where voters and  alternatives can be considered as  points in a metric space. Filos-Ratsikas and Voudouris \cite{filos2021approximate} studied a distributed single-facility location problem under the social cost objective in the discrete setting (i.e.,  facilities can only be located in a finite set of candidate locations) and the continuous setting (i.e., facilities can be located at every point on the real
line), respectively. For the discrete setting, they proved a tight bound of $7$ for strategyproof distributed mechanism, and for the continuous setting,  they proved a tight bound of $3$ for  strategyproof distributed  mechanism.
Further, Filos-Ratsikas et al. \cite{filos2024distortion} considered a continuous distributed single-facility location problem under four social objectives: the average cost, the max cost, the average-of-max cost and the max-of-average cost. A summary of their results is shown in Table 2.
\begin{table}\centering
\caption{Tight bounds on the  distortion of strategyproof distributed mechanisms in Filos-Ratsikas et al. \cite{filos2024distortion}.}\label{tab1}
\begin{tabular}{|l|c|c|}
\hline
Social objective &  Tight bounds \\
\hline
Average cost &  $3$ \cite{filos2024distortion} \\
Max cost &  $2$ \cite{filos2024distortion} \\
Max-of-Average cost&$1+\sqrt{2}$ \cite{filos2024distortion} \\
Average-of-Max cost& $1+\sqrt{2}$ \cite{filos2024distortion} \\
\hline
\end{tabular}
\end{table}

\section{Preliminaries}

Let $N=\left\{1,2,...,n\right\}$ be a set of agents. Each agent $i\in N$ has a private location $x_i\in \mathbb{R}$ and denote by $\mathbf{x}=(x_1,...,x_n)\in \mathbb{R}^n$ the location profile of all agents. The agents are divided into $k$ ($\ge 1$) disjoint groups and  let $D=\left\{1,...,k\right\}$ be the set of groups. For each group $d\in D$, let $N_d$ be the set of agents that belong to $d$  and $n_d=|N_d|$ the number of agents in group $d$.

Let $\mathcal F=\{ F_1,F_2\}$ be the two heterogeneous facilities to be located and $A=\left\{a_1,...,a_m\right\}\in \mathbb{R}^m$ be a multiset of candidate locations. Assume that  at most one facility can be located at each location in $A$. Denote by $I=(N,\mathbf{x},D,A)$ an instance. 

\noindent\textbf{Distributed Mechanism}. A distributed mechanism $M$ is a function that  maps an instance $I$ to a facility location profile $\mathbf{w}$ through two steps, i.e. $M(I)=\mathbf{w}=(w_1,w_2)$, where $w_1\in A$ is the location of $F_1$ and $w_2\in A\setminus\left\{w_1\right\}$ is the location of $F_2$. In detail, given an instance $I$, a distributed mechanism $M$ consists of two steps: 
\begin{itemize}
    \item[$\bullet$] Step $1$. For each group $d\in D$, $M$ chooses two representative locations $y_{d(1)}\in A$, $y_{d(2)}\in A\setminus \left\{y_{d(1)}\right\}$  for group $d$ based on the locations reported by agents in $N_d$.
    \item[$\bullet$] Step $2$. $M$ outputs a facility location profile $\mathbf{w}=(w_1,w_2)$ where $w_1\in A$, $w_2\in A\setminus \{w_1\}$  among all the representatives $\{y_{d{(1)}},y_{d{(2)}}\}_{d \in D}$.
\end{itemize}

\begin{remark}
  Following the work of Filos-Ratsikas et al. \cite{filos2024distortion}, a distributed mechanisms should have the following properties. \textit{(P1)}: For any two groups where the locations reported by agents are identical, the mechanism will output the same representative locations. \textit{(P2)}: The selected representative locations for a group are independent of the locations reported by agents in other groups, as well as  the number  and sizes of the other groups. \textit{(P3)}: The facility location profile $\mathbf{w}$ chosen by the mechanism is the same over all instances where the group representatives are identical. These properties are necessary for our work of the lower bounds.
\end{remark}
For any two points $x,y\in \mathbb{R}$, let $\delta(x,y)=|x-y|$ be the distance between $x$ and $y$. In our model, we assume that all agents approve the two heterogeneous facilities and the \textit{individual cost} of each agent $i$ for a facility location profile $\mathbf{w}$ is the distance from her location to the farthest facility:
\[c(x_i,\mathbf{w})=\max\left\{\delta(x_i,w_1),\delta(x_i,w_2)\right\},
\]

\noindent which is called the max-variant cost.

\noindent\textbf{Social Objectives}. We consider the following four  cost-minimization social objectives:

(1) The \textit{Average-of-Average cost} of a facility location profile $\mathbf{w}$ is the average of the average individual cost of agents in each group:  
\[  
\text{AoA}(\mathbf{w}|I) = \frac{1}{k}\sum_{d \in D}\left\{\frac{1}{n_d} \sum_{i \in N_d} c(x_i, \mathbf{w}) \right\} .
\]  

(2) The \textit{Max-of-Max cost} of a facility location profile  $\mathbf{w}$ is the maximum individual cost among all agents:  
\[  
\text{MoM}(\mathbf{w}|I) = \max_{d \in D} \max_{i \in N_d} c(x_i, \mathbf{w})  .
\]

(3) The \textit{Max-of-Average cost} of a facility location profile  $\mathbf{w}$ is the maximum of the average individual cost of all agents in each group:  
\[  
\text{MoA}(\mathbf{w}|I) = \max_{d \in D} \left\{\frac{1}{n_d}\sum_{i \in N_d} c(x_i, \mathbf{w})\right\}.
\]  

(4) The \textit{Average-of-Max cost} of a facility location profile  $\mathbf{w}$ is the average of the maximum individual cost among all agents in each group:  
\[  
\text{AoM}(\mathbf{w}|I) = \frac{1}{k}\sum_{d \in D} \max_{i \in N_d} c(x_i, \mathbf{w})  .
\]

\noindent\textbf{Strategyproofness.} In order to minimize individual costs, selfish agents may misreport their locations. Thus, the strategyproofness of mechanisms should be taken into account.    Formally, a mechanism $M$ is \textit{strategyproof }if each agent can never benefit by misreporting her location, regardless of the locations reported by the other agents, i.e., for every $i\in N$, for every $\mathbf{x}\in \mathbb{R}^n$,  and for every $x_i'\in \mathbb{R}$, it must hold that

\[  
c(x_i, M(N, (x_i, \mathbf{x}_{-i}), D, A) \leq c(x_i, M(N, (x_i', \mathbf{x}_{-i}), D, A))  ,
\]  

\noindent where \( \mathbf{x}_{-i} = (x_1, \ldots, x_{i-1}, x_{i+1}, \ldots, x_n) \) is the location profile of  $N\setminus\{i\}$.  

\noindent\textbf{Distortion.} In the distributed setting,  due to lack of global information and the requirement for strategyproofness of a mechanism, the facility location profile chosen by a distributed mechanism may be  suboptimal. Here, we adopt the notion of distortion to quantify the gap between a mechanism and the optimal mechanism. The \textit{distortion} of a  distributed mechanism $M$ is the worst-case  ratio between the social objective value obtained by $M$ and the optimal social objective value  over all possible instances:

\[  
\text{dist}(M) = \sup_{I} \frac{\text{cost}(M(I) \mid I)}{\text{cost}(\text{OPT}(I) \mid I)}  ,
\] 

\noindent where \(\text{OPT}(I)\) is the optimal solution for instance \(I\) and cost $\in\{AoA, MoM, MoA, AoM\}$. To simplify our notation,  $\text{cost}(M(I)|I)$ is abbreviated as $\text{cost}(M|I)$.

 For each agent $i\in N$, denote by $t({i})$ the closest candidate location to $i$ in $A$ and \(s(i) \) the closest candidate location to $i$ in $A\setminus \{t(i)\}$. Now, we give a class of distributed mechanisms which is called the  $(\alpha,\beta)$-Quantile mechanism ($\alpha,\beta\in [0,1]$).
 ~\\
 
\noindent\textit{\textbf{$\bm{(\alpha,\beta)}$-Quantile Mechanism}}\footnote{It is  regulated that when $\alpha=0$, the $\lceil\alpha\cdot n_d\rceil$-th leftmost location in $N_d$ is the leftmost location in $N_d$. Similarly, when $\beta=0$, the $\lceil\beta\cdot k\rceil$-th leftmost location is the leftmost location.}
\begin{itemize}
    \item[$\bullet$] Step 1. For each group $d\in D$, denote by $\alpha_d$ the $\lceil\alpha\cdot n_d\rceil$-th leftmost agent in $N_d$.
  Let $y_{d(1)}=t({\alpha_d})$ and $y _{d(2)}=s({\alpha_d})$ be two representatives of group $d$.
  \item[$\bullet$] Step 2. Set $z_d:=y_{d(1)}$ for each group  $d\in D$ and return $w_1:=$ the $\lceil\beta\cdot k\rceil$-th leftmost location in $\{z_d\}_{d\in D}$. For each group $d\in D$, update $z_d$ as $y_{d(2)}$ if $y_{d(1)}=w_1$, and return $w_2:=$ the $\lceil\beta\cdot k\rceil$-th leftmost location in $\{z_d\}_{d\in D}$.
\end{itemize}

\begin{remark}

Obviously, the $(\alpha,\beta)$-Quantile mechanism has the following properties.  For each group $d$, $y_{d(1)}$, $y_{d(2)}$ are two adjacent locations, which implies that $(y_{d(1)},y_{d(2)})$ and $(y_{d'(1)},y_{d'(2)})$ are never \lq\lq interlaced \rq\rq \footnote{Here, interlaced means the cases such as $y_{d(1)}<y_{d'(1)}<y_{d(2)}<y_{d'(2)}$.} for any two groups $d$ and $d'$. Therefore, the locations $w_1$, $w_2$ output by the  $(\alpha,\beta)$-Quantile mechanism must  serve as two representatives of some group $d$ such that $w_1=y_{d(1)}$, $w_2=y_{d(2)}$.
    
\end{remark}

We will prove that the $(\alpha,\beta)$-Quantile mechanism is strategyproof.

\begin{theorem}
The $(\alpha,\beta)$-Quantile mechanism is strategyproof.
\end{theorem}
\begin{proof}
Consider any instance $I$, and let $\mathbf{w}=(w_1,w_2)$ be the facility location profile output by the mechanism. According to the
property of the mechanism, $w_1$ and $w_2$ must be two adjacent locations closest to some agent $j$, i.e., $w_1=t(x_j)$, $w_2=s(x_j)$.
  Let $i$ be any agent  belonging to some group $d$.   Assume w.l.o.g. that $x_i\le x_{\alpha_d}$.

 Case 1: $x_i\le x_{\alpha_d}\le x_j$. In order to affect the output of the mechanism, agent $i$ must first become the $\lceil\alpha\cdot n_d\rceil$-th leftmost agent in $N_d$. So $i$ has to report a location $x_i' > x_{\alpha_d}$.  If $x_j\ge x_i' > x_{\alpha_d}$, then the output of the mechanism does not change, in which case $i$ has no incentive to misreport. If $x_i' > x_j$, then the output of the mechanism  becomes two adjacent  locations  closest to some agent $k$ with $x_k\ge x_j$, meaning that the cost of agent $i$ does not decrease. So $i$ has no incentive to misreport.

Case 2: $x_i<x_j< x_{\alpha_d}$. In order to affect the output of the mechanism, agent $i$ must first become the $\lceil\alpha\cdot n_d\rceil$-th leftmost agent in $N_d$. So $i$ has to report a location $x_i' > x_{\alpha_d}$. However, the output of the mechanism does not change, and then $i$ has no incentive to misreport.

Case 3: $x_j<x_i\le  x_{\alpha_d}$. Similar to Case 1, $i$ has no incentive to misreport.

 Above all, the $(\alpha,\beta)$-Quantile mechanism is strategyproof.

\qed
\end{proof}

In the following sections, we analyze the upper bounds on the distortion under four social objectives by adjusting parameters $\alpha$ and $\beta$ in the $(\alpha,\beta)$-Quantile mechanism.

\noindent\textbf{Notations.} Denote by \( l \), \( r \), and \( m \) the leftmost, rightmost, and median, respectively,  agent in \( N \) and  \( l_d, r_d, m_d \) the leftmost, rightmost, and median, respectively,  agent in \( N_d \).  For a facility location profile  $\mathbf{w}=(w_1,w_2)$ and  any $i\in N$, denote by $w(x_i)$ the farthest one to agent $i$ in $\{w_1,w_2\}$.

\section{Average-of-Average cost}

In this section, we consider the Average-of-Average cost objective, which is the average of the average individual cost of agents in each group.  For the upper bound, we consider the  $(\frac{1}{2},\frac{1}{2})$-Quantile mechanism,\footnote{Here, we use the method of undetermined coefficients to conclude that the $(\frac{1}{2},\frac{1}{2})$-Quantile mechanism achieves a minimum upper bound on the distortion in  $(\alpha,\beta)$-Quantile.} which achieves a distortion of at most $9$. For the lower bound, the distortion of any strategyproof mechanism is at least $3-\epsilon$, for any $\epsilon>0$.
~\\

\noindent \textbf{$\bm{(\frac{1}{2},\frac{1}{2})}$-Quantile Mechanism} 
\begin{itemize}
    \item[$\bullet$] Step 1. For each group $d\in D$ , let $y_{d(1)}=t({m_d})$, $y _{d(2)}=s({m_d})$.
    \item[$\bullet$]  Step 2. Set $z_d:=y_{d(1)}$ for each group $d\in D$ and return $w_1:=$ the median location in $\{z_d\}_{d\in D}$. For each group $d\in D$, update $z_d$ as $y_{d(2)}$ if $y_{d(1)}=w_1$, and return $w_2:=$ the median location in $\{z_d\}_{d\in D}$.
\end{itemize}

\begin{theorem}
For the Average-of-Average cost,\footnote{Due to space constrains,  all missing proofs can be found in the appendix.} the distortion of the $(\frac{1}{2},\frac{1}{2})$-Quantile mechanism is at most $9$.
\end{theorem}

\begin{theorem}
For the Average-of-Average cost, the distortion of any strategyproof mechanism is at least $3-\epsilon$, for any $\epsilon>0$.
\end{theorem}

\begin{proof}
Assume for contradiction  that there exists a strategyproof mechanism $M$ that has a distortion strictly smaller than $3-\epsilon$, for some $\epsilon>0$.
 Let $\theta>0$ be an infinitesimal.  We will reach contradictions through considering the following several instances with the set of candidate locations $A=\left\{0,0,1,1\right\}$.

\textbf{Instance $I_1$}: Consider an instance $I_1$ with one agent in a single group, where $x_1^1=0$. In this case, the representative chosen by mechanism $M$ must be $(0,0)$, i.e., $M(I_1)=(0,0)$. Otherwise, the distortion of $M$ is infinite.

\textbf{Instance $I_2$}: Consider an instance $I_2$ with one agent in a single group, where $x_1^2=\frac{1}{2}-\theta$. If $M(I_2)\ne (0,0)$, then $c(x_1^2,M(I_2))=\frac{1}{2}+\theta$. Considering $c(x_1^2,M(I_1))=\frac{1}{2}-\theta$,  agent 1 at $x_1^2$ can decrease her cost by misreporting her location as $0$, in contradiction to strategyproofness. Thus, $M(I_2)=(0,0)$.

\textbf{Instance $I_3$}: Consider an instance $I_3$ with one agent in a single group, where $x_1^3=1$. Similar to Instance $I_1$, mechanism $M$ must output  $(1,1)$ as the representative of the group, i.e., $M(I_3)=(1,1)$.

\textbf{Instance $I_4$}: Consider an instance $I_4$ with one agent in a single group, where $x_1^4=\frac{1}{2}+\theta$. Similar to Instance $I_2$, $M(I_4)=(1,1)$.

\textbf{Instance $I_5$}: Consider an Instance $I_5$ with two groups:

\begin{itemize}  
    \item In group $1$, there is one agent at $\frac{1}{2}-\theta$.
    \item In group $2$, there is one agent at $1$.
\end{itemize} 

\noindent \noindent Considering $M(I_2)=(0,0)$ and $M(I_3)=(1,1)$, by \textit{(P1)} and \textit{(P2)}, the representatives of group $1$ and group $2$ in Instance $I_5$ are \( (0,0) \) and \( (1,1) \), respectively.  Since $AoA(0,0)=\frac{3}{4}-\frac{\theta}{2}$, $AoA(0,1)=AoA(1,0)=\frac{3}{4}+\frac{\theta}{2}$ and $AoA(1,1)=\frac{1}{4}+\frac{\theta}{2}$, $M$ must output $(1,1)$ as the overall facility location profile, i.e., $M(I_5)=(1,1)$. Otherwise, the distortion of $M$ is at least $\frac{\frac{3}{2}-\theta}{\frac{1}{2}+\theta}\ge3-\epsilon$, for $\theta\le\frac{\epsilon}{8-2\epsilon}$.

 Now, we can reach a contradiction by considering the following instance $I_6$ with two groups:

\begin{itemize}  
    \item In group $1$, there is one agent at $\frac{1}{2}+\theta$.  
    \item In group $2$, there is one agent at $0$.
\end{itemize} 

\noindent Considering $M(I_4)=(1,1)$ and $M(I_1)=(0,0)$, by \textit{(P1)} and \textit{(P2)}, the representatives of group $1$ and group $2$ in Instance $I_6$  are \( (1,1) \) and \( (0,0) \), respectively. Since the group representatives in Instance $I_5$ and Instance $I_6$ are identical,  according to \textit{(P3)}, we have that $M(I_6)=(1,1)$.
Since $AoA(0,0)=\frac{1}{4}+\frac{\theta}{2}$, $AoA(0,1)=AoA(1,0)=\frac{3}{4}+\frac{\theta}{2}$ and $AoA(1,1)=\frac{3}{4}-\frac{\theta}{2}$, the distortion of $M$ is at least $\frac{\frac{3}{2}-\theta}{\frac{1}{2}+\theta}\ge3-\epsilon$, for $\theta\le\frac{\epsilon}{8-2\epsilon}$. This contradicts the assumption that M has a distortion strictly smaller than $3-\epsilon$.
\qed
\end{proof}

\begin{remark}
  In the  facility location problem,  the most natural social objective is the social cost. In fact, our results under the average-of-average cost can be easily generalized to  the social cost.
\end{remark}

\section{Max-of-Max cost}

In this section, we study the Max-of-Max cost objective. For the upper bound, we consider the $(1,1)$-Quantile mechanism,\footnote{In fact, whatever values $\alpha$ and $\beta$ take, the $(\alpha,\beta)$-Quantile mechanism would achieve a distortion of at most $3$ under the max-of-max cost.} which achieves a distortion of at most $3$.
~\\

\noindent \textbf{$\bm{(1,1)}$-Quantile Mechanism}  
\begin{itemize}
    \item[$\bullet$] Step 1. For each group $d\in D$ , let $y_{d(1)}=t({r_d})$, $y _{d(2)}=s({r_d})$.
    \item[$\bullet$] Step 2. Set $z_d:=y_{d(1)}$ for each group  $d\in D$ and return $w_1:=$ the rightmost location in $\{z_d\}_{d\in D}$. For each group $d\in D$, update $z_d$ as $y_{d(2)}$ if $y_{d(1)}=w_1$, and return $w_2:=$ the rightmost location in $\{z_d\}_{d\in D}$.
\end{itemize} 
\begin{theorem}
For the Max-of-Max cost, the distortion of the $(1,1)$-Quantile mechanism is at most 3. 
\end{theorem}

\begin{proof}
	Given any instance $I$, let $\mathbf{w} = (w_1,w_2)$ be the location profile chosen by the mechanism and $\mathbf{o}= (o_1,o_2)$ be an optimal solution. According to the property of this mechanism, we have that $w_1 = t(r)$ and $w_2 = s(r)$. Denote by $i^*$ the agent  such that $MoM(\mathbf{w}) = \max\{\delta(x_{i^*}, w_1), \delta(x_{i^*}, w_2)\} = \delta(x_{i^*}, w(x_{i^*}))$. Since $w_1=t(r)$ and $w_2=s(r)$ are the closest candidate locations to $r$ and using the triangle inequality, we have that  
	\begin{equation}
		\begin{aligned}
			MoM(\mathbf{w}) &= \delta(x_{i^*}, w(x_{i^*})) \\
			&\leq \delta(x_{i^*}, o_1) + \delta(o_1, x_r) + \delta(x_r, w(x_{i^*}))\\
			&\le \delta(x_{i^*}, o(x_{i^*})) + \delta(x_r, o(x_r)) + \delta(x_r, w(x_r)) \\
			&\le \delta(x_{i^*}, o(x_{i^*})) + 2 \cdot \delta(x_r, o(x_r)).   
		\end{aligned}
	\end{equation}
	
	\noindent For the optimal solution $\mathbf{o}$, we have 
	$MoM(\mathbf{o}) \geq \delta(x_i, o(x_i))$, for any agent $i\in N$.  
	Therefore, we obtain
	\begin{equation}
		\begin{aligned}
			MoM(\mathbf{w})&\leq \delta(x_{i^*},o(x_{i^*}))+2\cdot \delta(x_r,o(x_r))\leq3\cdot MoM(\mathbf{o}).
		\end{aligned}
	\end{equation}
    \qed
\end{proof}
 Next, we show a lower bound on the distortion of all strategyproof mechanisms.

\begin{theorem}
For the Max-of-Max cost, the distortion of any strategyproof mechanism is at least $3-\varepsilon$, for any $\varepsilon>0$.
\end{theorem}

\begin{proof}
Assume for contradiction that there exists a strategyproof mechanism $M$
that has a distortion strictly smaller than $3-\varepsilon$, for some $\varepsilon>0$.   Let $\theta>0$ be an infinitesimal.  We will reach a contradiction through considering the following several instances with the set of candidate locations $A=\left\{0,0,2,2,4,4\right\}$.

\textbf{Instance $I_1$}: Consider an instance $I_1$ with two agents in a single group, with the location profile $\mathbf{x^1}$,  where $x_1^1=-1$ and $x_2^1=1-\theta$. We can obtain $MoM (0,0) = 1, MoM (0,2) = MoM (2,0) = MoM (2,2) = 3, MoM (4,4) = MoM (0,4) = MoM (4,0) = MoM (2,4) = MoM (4,2)=5 $. Clearly, $OPT(I_1)=(0,0)$, so we claim that $M$ must choose $(0,0)$ as the representative of this group, i.e., $M(I_1)=(0,0)$. Otherwise, the distortion of $M$ is at least $3$, which contradicts the assumption.

\textbf{Instance $I_2$}: Consider an instance $I_2$ with two agents in a single group, with the location profile $\mathbf{x^2}$, where $x_1^2=1-\theta$ and $x_2^2=1-\theta$. If $M(I_2)\ne (0,0)$, then $c(x_1^2,M(I_2))\ge 1+\theta$. Considering $c(x_1^2,M(I_1))=1-\theta$,
  agent $1$ at $x_1^2$ can decrease her cost by misreporting her location as $-1$, in contradiction to strategyproofness. Thus, $M(I_2)=(0,0)$.

\textbf{Instance $I_3$}: Consider an instance $I_3$  with two agents in a single group, with the location profile $\mathbf{x^3}$, where  $x_1^3=5$ and $x_2^3=3+\theta$. Since Instance $I_3$ and Instance $I_1$ have symmetry with respect to $A$, we have that $M(I_3)=(4,4)$.

\textbf{Instance $I_4$}: Consider an instance $I_4$ with two agents in a single group, with the location profile $\mathbf{x^4}$, where $x_1^4=3+\theta$ and $x_2^4=3+\theta$. Similar to Instance $I_2$, we have that  $M(I_4)=(4,4)$.

 Finally, to reach a contradiction, consider the following instance $I_5$ with two groups:

\begin{itemize}  
    \item In group $1$, there are two agents at \( 1 - \theta \).  
    \item In group $2$, there are two agents at \( 3 + \theta \).  
\end{itemize}  
\noindent Considering $M(I_2)=(0,0)$ and $M(I_4)=(4,4)$, by \textit{(P1)} and \textit{(P2)}, the representatives of group $1$ and group $2$ 
 are \( (0,0) \) and \( (4,4) \), respectively. Therefore, \( M(I_5) = (0,0), (0,4), (4,0) \) or \( (4,4) \) and  $MoM(M|I_5) = 3 + \theta$.
However, \( OPT(I_5) = (2,2) \) and   $MoM(OPT|I_5) = 1 + \theta$.

Thus, the distortion of mechanism \( M \) is at least   
$\frac{3 + \theta}{1 + \theta} \geq 3 - \varepsilon$,  for $\theta\le \frac{\epsilon}{2-\epsilon}$; which is a contradiction.  
\qed
\end{proof}

\section{Max-of-Average cost}

In this section, we focus on the Max-of-Average cost objective, which is the maximum of the average individual cost of all agents in each group.
For the upper bound, we consider the  $(\frac{3-\sqrt{5}}{2},1)$-Quantile mechanism, which achieves a distortion of at most $2+\sqrt{5}$. For the lower bound, the distortion of any strategyproof mechanism is at least $\frac{7}{2}-\epsilon$, for any $\epsilon>0$.

 We first consider  the $(\alpha,1)$-Quantile mechanism and analyze its performance.
~\\

\noindent \textbf{$\bm{(\alpha,1)}$-Quantile Mechanism} 

\begin{itemize}
    \item[$\bullet$] Step 1. For each group $d\in D$ , denote by $\alpha_d$ the $\lceil\alpha\cdot n_d\rceil$-th leftmost agent in $N_d$.  Let $y_{d(1)}=t({\alpha_d})$, $y _{d(2)}=s({\alpha_d})$.
    \item[$\bullet$]  Step 2. Set $z_d:=y_{d(1)}$ for each group  $d\in D$ and return $w_1:=$ the rightmost location in $\{z_d\}_{d\in D}$. For each group $d\in D$, update $z_d$ as $y_{d(2)}$ if $y_{d(1)}=w_1$, and return $w_2:=$ the rightmost location in $\{z_d\}_{d\in D}$.
\end{itemize} 
 
\begin{lemma}
    For the Max-of-Average cost, the distortion of the $(\alpha,1)$-Quantile mechanism is at most $\max \left\{ 1+\frac{2(1-\alpha)}{\alpha}, 1 + \frac{2}{1 - \alpha}\right\}$.
\end{lemma}

By solving the equation $1 + \frac{2(1 - \alpha)}{\alpha} = 1 + \frac{2}{1 - \alpha}$, we can obtain the following theorem.

\begin{theorem}
For the Max-of-Average cost, the distortion of the $(\frac{3 - \sqrt{5}}{2},1)$-Quantile mechanism is at most \( 2 + \sqrt{5} \).  
\end{theorem}

\begin{theorem}
For the Max-of-Average cost, the distortion of any strategyproof mechanism is at least $\frac{7}{2}-\varepsilon$, for any $\varepsilon>0$.
\end{theorem}

\section{Average-of-Max cost}

In this section, we turn our attention to the last social objective, Average-of-Max cost objective, which is the average of the maximum individual cost among all agents in each group. For the upper bound, we consider the $(1,\frac{3-\sqrt{5}}{2})$-Quantile mechanism, which achieves a distortion of at most $2+\sqrt{5}$ . For the lower bound, the distortion of any strategyproof mechanism is at least $3-\epsilon$, for any $\epsilon>0$.

We first consider the $(1,\beta)$-Quantile mechanism and analyze its performance.
~\\

\noindent \textbf{$\bm{(1,\beta)}$-Quantile Mechanism} 

\begin{itemize}
\item[$\bullet$] Step 1. For each group $d\in D$ , let $y_{d(1)}=t({r_d})$, $y _{d(2)}=s({r_d})$.

\item[$\bullet$] Step 2. Set $z_d:=y_{d(1)}$ for each group  $d\in D$ and return $w_1:=$ the $\lceil\beta\cdot k\rceil$-th  leftmost location in $\{z_d\}_{d\in D}$. For each group  $d\in D$, update $z_d$ as $y_{d(2)}$ if $y_{d(1)}=w_1$, and return $w_2:=$  the $\lceil\beta\cdot k\rceil$-th  leftmost location in $\{z_d\}_{d\in D}$.
\end{itemize} 

\begin{lemma}
    For the Average-of-Max cost, the distortion of the $(1,\beta)$-Quantile mechanism is at most $\max \left\{ 1+\frac{2(1-\beta)}{\beta}, 1 + \frac{2}{1 - \beta}\right\}$.
\end{lemma}

By solving the equation $1 + \frac{2(1 - \beta)}{\beta} = 1 + \frac{2}{1 - \beta}$, we can obtain the following theorem. 
\begin{theorem}
For the Average-of-Max cost, the distortion of the  $(1,\frac{3-\sqrt{5}}{2})$-Quantile mechanism  is at most $2+\sqrt{5}$.
\end{theorem}

\begin{theorem}
For the Average-of-Max cost, the distortion of any strategyproof mechanism is at least $3-\epsilon$, for any $\epsilon>0$.
\end{theorem}

\section{Conclusions and Future Work}

In this paper, we studied a constrained distributed heterogeneous two-facility location problem and showed upper and lower bounds on the distortion of strategyproof mechanisms under four social objectives. There are at least three directions for future work. First, it would be interesting to close the gaps between our lower and upper bounds.  Second, it would be meaningful to consider more than just two facilities to be built and more general metric spaces than just the real line. Third, one could try some other social objectives in  distributed settings (such as minimax envy), and consider  settings where agents have other preferences over the facilities (such as the dual preferences).

\begin{credits}
    \subsubsection{\ackname} This research was supported in part by the National Natural Science Foundation of China (12201590, 12171444) and Natural Science Foundation of Shandong Province (ZR2024MA031).

\subsubsection{\discintname}
\begin{footnotesize}
     The authors declare that they have no conflict of interest.
\end{footnotesize}

\end{credits}

\bibliographystyle{splncs04}
\bibliography{main.bib}
%






\begin{appendix}
\section{Appendix}

\paragraph{Proof of Theorem 2}

    Given any instance $I$ , let $\mathbf{w}=(w_1,w_2)$ be the location profile chosen by the mechanism and $\mathbf{o}=(o_1,o_2)$ be an optimal solution. Without loss of generality, we assume that $w_1\le w_2$ and $o_1\le o_2$. According to the property of this mechanism, there exists a group $d^*$ such that \( w_1 = y_{d^*(1)} = t({m_{d^*}}) \) and $w_2 = y_{d^*(2)} =s(m_{d^*})$. For each group $d\in D$, let $N_{d(1)}\subseteq N_{d}$ be the set of agents in $N_{d}$ where $w(x_i)=w_1$ and  $N_{d(2)}\subseteq N_{d}$ be the set of agents in $N_{d}$ where $w(x_i)=w_2$.  Using the triangle inequality, we can obtain
    \begin{equation}
     \begin{aligned}
  AoA(\mathbf{w}) &= \frac{1}{k}\sum_{d \in D} \left\{ \frac{1}{n_d}\sum_{i \in N_{d(1)}} \delta(x_i, w_1) +\frac{1}{n_d} \sum_{i \in N_{d(2)}} \delta(x_i, w_2) \right\}\\
  & \leq\frac{1}{k} \sum_{d \in D} \left\{\frac{1}{n_d} \sum_{i \in N_{d(1)}} \left[ \delta(x_i, o_1) + \delta(o_1, w_1) \right] 
 + \frac{1}{n_d}\sum_{i \in N_{d(2)}} \left[ \delta(x_i, o_2) + \delta(o_2, w_2) \right] \right\}\\
 & \leq\frac{1}{k} \sum_{d \in D} \left\{\frac{1}{n_d} \sum_{i \in N_{d(1)}} \delta(x_i, o(x_i)) 
 + \frac{1}{n_d}\sum_{i \in N_{d(2)}}  \delta(x_i, o(x_i))  \right\}\\
 & +\frac{1}{k}\sum_{d \in D} \left\{\frac{1}{n_d}\sum_{i \in N_{d(1)}}  \delta(o_1, w_1) +\frac{1}{n_d}\sum_{i \in N_{d(2)}}  \delta(o_2, w_2) \right\}\\
 &= AoA(\mathbf{o})+\frac{1}{k}\sum_{d \in D} \left\{\frac{1}{n_d}\sum_{i \in N_{d(1)}}  \delta(o_1, w_1) +\frac{1}{n_d}\sum_{i \in N_{d(2)}}  \delta(o_2, w_2) \right\}.
\end{aligned}
\end{equation}

\noindent\textbf{Case 1:} \( o_1 \leq o_2 < w_1 \leq w_2 \). Since $\delta(o_1,w_1)\le \delta(o_1,w_2)$ and $\delta(o_2,w_2)\le \delta(o_1,w_2)$, we have
\begin{equation}
     \begin{aligned}
  AoA(\mathbf{w}) &\le AoA(\mathbf{o})+\frac{1}{k}\sum_{d \in D} \left\{\frac{1}{n_d}\sum_{i \in N_{d(1)}}  \delta(o_1, w_1) +\frac{1}{n_d}\sum_{i \in N_{d(1)}}  \delta(o_2, w_2) \right\}\\
  &\le AoA(\mathbf{o})+\delta(o_1,w_2).
\end{aligned}
\end{equation}
Let $S=\{ d\in D\mid y_{d(1)}\ge w_1, y_{d(2)}\ge w_1\}$.
 For each group $d\in D$, since $y_{d(1)}$, $y_{d(2)}$ are adjacent and, \( w_1 \) is the median location in \( \{ z_{d} \}_{d \in D}=\{ y_{d(1)} \}_{d \in D} \) and \( w_2 \) is the median location in the updated \( \{ z_{d} \}_{d \in D} \), we have that \( |S| \geq \frac{1}{2} k \).
For each group $d\in S$, let \( T_d = \{ i \in N_d \mid x_i \geq x_{m_d} \} \). By the definition of $m_d$, $|T_d|\ge \frac{1}{2} n_d$.
Since \( w_1 \) and \( w_2 \) are the closest locations to \( m_{d^*} \), it holds that \( x_{m_{d^*}} \ge  \frac{o_2 + w_2}{2} \). Then, for any \( d \in S \), $i\in T_d$, we have $\delta(x_i, o_1) \geq \delta(o_1, o_2) + \frac{\delta(o_2, w_2)}{2}\ge \frac{\delta(o_1,w_2)}{2}$. Hence,
\begin{equation}
     \begin{aligned}
  AoA(\mathbf{o}) &\ge \frac{1}{k}\sum_{d \in D} \left\{\frac{1}{n_d}\sum_{i \in N_{d}}  \delta(x_i, o_1) \right\}\ge \frac{1}{k}\sum_{d \in S} \left\{\frac{1}{n_d}\sum_{i \in T_{d}}  \delta(x_i, o_1) \right\}\\
  &\ge \frac{1}{k}\sum_{d \in S} \left\{\frac{1}{n_d}\sum_{i \in T_{d}}  \frac{\delta(o_1,w_2)}{2} \right\}\ge \frac{\delta(o_1,w_2)}{8}.
\end{aligned}
\end{equation}

\noindent Therefore, by (4) and (5), we obtain
\begin{equation}
     \begin{aligned}
  AoA(\mathbf{w}) &\le AoA(\mathbf{o})+\delta(o_1,w_2)\le 9\cdot AoA(\mathbf{o}).
\end{aligned}
\end{equation}

\noindent\textbf{Case 2:} \( o_1 < o_2 = w_1< w_2 \). Similar to Case 1, we can obtain $AoA(\mathbf{w}) \le 9\cdot AoA(\mathbf{o}).$

\noindent\textbf{Case 3:} \( w_1 \le w_2 < o_1 \le o_2 \).  In this case, since $\delta(o_1,w_1)\le \delta(o_2,w_1)$ and $\delta(o_2,w_2)\le \delta(o_2,w_1)$, we have
\begin{equation}
     \begin{aligned}
  AoA(\mathbf{w}) \le AoA(\mathbf{o})+\delta(o_2,w_1).
\end{aligned}
\end{equation}
\noindent  Similar to Case 1, let \( S' = \{ d \in D \mid y_{d(1)} \le w_1\}  \). Since \( w_1 \) is the median location in \( \{ y_{d(1)} \}_{d \in D} \), we have that \( |S'| \geq \frac{1}{2} k \). For each group $d\in S'$, let \( T'_d = \{ i \in N_d \mid x_i \le x_{m_d} \} \). By the definition of $m_d$, $|T'_d|\ge \frac{1}{2} n_d$.
Since \( w_1 \) and \( w_2 \) are the closest  locations to \( m_{d^*} \), it holds that \( x_{m_{d^*}}\le \frac{w_1+w_2}{2} \). Then, for any \( d \in S' \), $i\in T'_d$, we have $\delta(x_i, o_2) \geq \delta(o_2, w_2) + \frac{\delta(w_1, w_2)}{2}\ge \frac{\delta(o_2,w_1)}{2}$. Hence,
\begin{equation}
     \begin{aligned}
  AoA(\mathbf{o}) \ge \frac{1}{k}\sum_{d \in S'} \left\{\frac{1}{n_d}\sum_{i \in T'_{d}}  \delta(x_i, o_2) \right\}\ge \frac{1}{k}\sum_{d \in S'} \left\{\frac{1}{n_d}\sum_{i \in T'_{d}}  \frac{\delta(o_2,w_1)}{2} \right\}\ge \frac{\delta(o_2,w_1)}{8}.
\end{aligned}
\end{equation}
Therefore, by (7) and (8), we obtain
\begin{equation}
     \begin{aligned}
  AoA(\mathbf{w}) &\le AoA(\mathbf{o})+\delta(o_2,w_1)\le 9\cdot AoA(\mathbf{o}).
\end{aligned}
\end{equation}
\noindent\textbf{Case 4:}  $w_1 <w_2 = o_1\le o_2$.  Similar to Case 3, we can obtain ${AoA}(\mathbf{w})\le 9\cdot AoA(\mathbf{o}).$

\noindent\textbf{Case 5:} \( o_1 <  w_1 \le w_2 \le o_2 \).  In this case, we have that $AoA(\mathbf{w})\leq AoA(\mathbf{o}) + \delta(o_1, o_2). $

\noindent Clearly, 
\begin{equation}
    \begin{aligned}
       AoA(\mathbf{o}) \ge \frac{1}{k}\sum_{d \in D} \left\{\frac{1}{n_d}\sum_{i \in N_{d}}  \delta(x_i, o_1) \right\},
       AoA(\mathbf{o}) \ge \frac{1}{k}\sum_{d \in D} \left\{\frac{1}{n_d}\sum_{i \in N_{d}}  \delta(x_i, o_2) \right\}
\end{aligned}
\end{equation}
\noindent Using the triangle inequality, we have that
\begin{equation}
    \begin{aligned}
 AoA(\mathbf{o}) &\ge \frac{1}{2k}\sum_{d \in D} \left\{\frac{1}{n_d}\sum_{i \in N_{d}}  \left[\delta(x_i, o_1)+\delta(x_i,o_2)\right]\right\}\\
 &\ge \frac{1}{2k}\sum_{d \in D} \left\{\frac{1}{n_d}\sum_{i \in N_{d}}  \delta(o_1, o_2)\right\}\ge \frac{\delta(o_1,o_2)}{2}.
\end{aligned}
\end{equation}
\noindent Therefore, we can obtain $AoA(\mathbf{w})\le 3\cdot AoA(\mathbf{o}).$

\noindent\textbf{Case 6:} \(o_1= w_1\le w_2<o_2 \). Similar to Case 5, we can obtain $AoA(\mathbf{w})\le 3\cdot AoA(\mathbf{o}).$

\noindent Above all, 
\begin{equation}
    \begin{aligned} 
       AoA(\mathbf{w})\le 9\cdot AoA(\mathbf{o}).
\end{aligned}
\end{equation}
\qed

\paragraph{Proof of Lemma 1}

  Given any instance $I$, let \( \mathbf{w}= (w_1, w_2) \) be the location profile chosen by the mechanism and \( \mathbf{o} = (o_1, o_2) \) be an optimal solution. Assume w.l.o.g. that \( w_1 \leq w_2 \) and \( o_1 \leq o_2 \).  According to the property of this mechanism, there exists a group $d^*$ such that \( w_1 = y_{d^*(1)} = t({\alpha_{d^*}}) \) and $w_2 = y_{d^*(2)} =s(\alpha_{d^*})$.
Denote by \( d' \) a group with the maximum average of individual cost of agent for \( \mathbf{w} \), such that $MoA(\mathbf{w})=\frac{1}{n_{d'}}\sum_{i\in N_{d'}}\delta(x_i,w(x_i) )$.  Let $N_{d'(1)}\subseteq N_{d'}$ be the set of agents in $N_{d'}$ where $w(x_i)=w_1$ and  $N_{d'(2)}\subseteq N_{d'}$ be the set of agents in $N_{d'}$ where $w(x_i)=w_2$.

\noindent\textbf{Case 1:} \( o_1 \leq o_2 < w_1 \leq w_2 \). By the definition of \( d' \) and the triangle inequality, we have that
\begin{equation}
    \begin{aligned}
       MoA(\mathbf{w})
         &=\frac{1}{n_{d'}}\sum_{i \in N_{d'(1)} } \delta(x_i, w_1) + \frac{1}{n_{d'}}\sum_{i \in N_{d'(2)} } \delta(x_i, w_2)\\
             &\leq\frac{1}{n_{d'}} \sum_{i \in N_{d'} } \delta(x_i, o(x_i)) + \frac{1}{n_{d'}}\sum_{i \in N_{d'(1)} } \delta(o_1, w_1) + \frac{1}{n_{d'}}\sum_{i \in N_{d'(2)} } \delta(o_2, w_2)\\  
            &\leq MoA(\mathbf{o}) + \delta(o_1, w_2)  .
\end{aligned}
\end{equation}
\noindent Let $S=\left\{i\in N_{d^*}|x_i\ge x_{\alpha_{d^*}}\right\}$. By the definition of $\alpha_{d^*}$, $\left|S\right|\ge (1-\alpha)\cdot n_{d^*}$. Since \( o_1 \leq o_2 < w_1 \leq w_2 \) and $w_1$, $w_2$ are the closest  locations to $\alpha_{d^*}$, all agents in $S$ are closer to $\mathbf{w}=(w_1,w_2)$ than $\mathbf{o}=(o_1, o_2)$. So we have $\delta(x_i,o_2)\ge \frac{\delta(o_2,w_2)}{2}$ and $\delta(x_i,o_1)\ge \delta(o_1,o_2)+\frac{\delta(o_2,w_2)}{2}$, for any $i\in S$.
Using these properties, we can obtain:
\begin{equation}
    \begin{aligned}
       MoA(\mathbf{o})&\ge\frac{1}{n_{d^*}}\sum_{i \in N_{d^*} } \delta(x_i, o(x_i))\ge\frac{1}{n_{d^*}}\sum_{i \in S } \delta(x_i, o_1) \\
             &\ge(1-\alpha)\cdot\left[\frac{\delta(o_2,w_2)}{2}+\delta(o_1,o_2)\right],
\end{aligned}
\end{equation}

\noindent and
\begin{equation}
    \begin{aligned}
       MoA(\mathbf{o})\ge\frac{1}{n_{d^*}}\sum_{i \in S } \delta(x_i, o_2) 
            \ge(1-\alpha)\cdot\frac{\delta(o_2,w_2)}{2}.
\end{aligned}
\end{equation}

\noindent Then
\begin{equation}
    \begin{aligned} 
       MoA(\mathbf{o})&\ge\frac{(1-\alpha)}{2}\cdot\delta(o_1,w_2).
\end{aligned}
\end{equation}
\noindent Therefore, by (13) and (16), we obtain
\begin{equation}
    \begin{aligned} 
       MoA(\mathbf{w})&\le MoA(\mathbf{o})+\delta(o_1,w_2)\le \left(1+\frac{2}{1-\alpha}\right)\cdot MoA(\mathbf{o}).
\end{aligned}
\end{equation}

\noindent\textbf{Case 2:} \( o_1 <o_2 = w_1 < w_2 \). Similar to Case 1, we can obtain 
\begin{equation}
    \begin{aligned} 
       MoA(\mathbf{w})\le \left(1+\frac{2}{1-\alpha}\right)\cdot MoA (\mathbf{o}).
\end{aligned}
\end{equation}

\noindent\textbf{Case 3:} \( w_1\le w_2 <o_1 \le  o_2 \). Let $L$ be the set of agents in $N_{d'}$ from the first leftmost agent to $\alpha_{d'}$ and $R$ be the set of the remaining agents in $N_{d'}$.
 By the definition of $\alpha_{d'}$, we have that $\left|L\right|=\alpha\cdot n_{d'}$ and $\left|R\right|=(1-\alpha)\cdot n_{d'}$. As $w_1=\max\left\{y_{d(1)}\right\}_{d\in D}$, we can obtain $y_{d'(1)}\le w_1\le w_2<o_1\le o_2$. Since $y_{d'(1)}$ is the closest location to $\alpha_{d'}$, it holds that 
$\delta(x_i,w(x_i))\le \delta(x_i,o(x_i))$, for any $i\in L$. Using these properties, we can obtain
\begin{equation}
    \begin{aligned}
MoA(\mathbf{w}) &= \frac{1}{n_{d'}}\sum_{i \in L} \delta(x_i, w(x_i)) + \frac{1}{n_{d'}}\sum_{i \in R} \delta(x_i, w(x_i))\\ 
&\leq \frac{1}{n_{d'}}\sum_{i \in L} \delta(x_i, o(x_i)) +\frac{1}{n_{d'}} \sum_{i \in R \cap N_{d'(1)}} \delta(x_i,w_1)+ \frac{1}{n_{d'}}\sum_{i \in R \cap N_{d'(2)}} \delta(x_i,w_2)\\
&\leq \frac{1}{n_{d'}}\sum_{i \in L} \delta(x_i, o(x_i)) + \frac{1}{n_{d'}}\sum_{i \in R \cap N_{d'(1)}} [\delta(x_i,o_1)+\delta(o_1,w_1)]  \\ 
&+ \frac{1}{n_{d'}}\sum_{i \in R \cap N_{d'(2)}} [\delta(x_i, o_2) + \delta(o_2, w_2)]  \\ 
&\leq MoA(\mathbf{o}) + (1 - \alpha) \cdot \delta(o_2, w_1). 
\end{aligned}
\end{equation}

\noindent Since all agents in $L$ are closer to $\mathbf{w}=(w_1,w_2)$ than $\mathbf{o}=(o_1,o_2)$,  we have that $\delta(x_i,o_2)\ge \frac{\delta(w_1,w_2)}{2}+\delta(o_2,w_2)$ (when $w_1\ne w_2$) and $\delta(x_i,o_2)\ge \frac{\delta(o_1,w_2)}{2}+\delta(o_1,o_2)$ (when $w_1=w_2$), for any $i\in L$. Thus,
\begin{equation}
    \begin{aligned}
 MoA(\mathbf{o})&\ge\frac{1}{n_{d'}}\sum_{i \in N_{d'} } \delta(x_i, o(x_i))\ge\frac{1}{n_{d'}}\sum_{i \in L } \delta(x_i, o_2) \\
             &\ge\alpha\cdot\left[\frac{\delta(w_1,w_2)+\delta(o_2,w_2)}{2} \right]\ge \frac{\alpha}{2}\cdot\delta(o_2,w_1).
\end{aligned}
\end{equation}

\noindent Therefore, by (19) and (20), we obtain
\begin{equation}
    \begin{aligned} 
       MoA(\mathbf{w})\le MoA(\mathbf{o})+(1-\alpha)\cdot\delta(o_2,w_1)\le \left(1+\frac{2(1-\alpha)}{\alpha}\right)\cdot MoA(\mathbf{o}).
\end{aligned}
\end{equation}

\noindent\textbf{Case 4:} \(  w_1 <w_2=o_1\le o_2 \). Similar to Case 3, we can obtain $MoA(\mathbf{w})\le \left(1+\frac{2(1-\alpha)}{\alpha}\right)\cdot MoA(\mathbf{o}).$

\noindent\textbf{Case 5:} \(o_1< w_1\le w_2\le o_2 \). By the definition of $d'$ and the triangle inequality, we have that
\begin{equation}
    \begin{aligned}
       MoA(\mathbf{w})&=\frac{1}{n_{d'}}\sum_{i \in N_{d'(1)} } \delta(x_i, w_1) + \frac{1}{n_{d'}}\sum_{i \in N_{d'(2)} } \delta(x_i, w_2)\\
 &\leq MoA(\mathbf{o}) + \delta(o_1, o_2). 
\end{aligned}
\end{equation}

\noindent Clearly, $MoA(\mathbf{o})\ge \frac{1}{n_{d'}}\sum_{i \in N_{d'} } \delta(x_i, o_1)$ and $MoA(\mathbf{o})\ge\frac{1}{n_{d'}} \sum_{i \in N_{d'} } \delta(x_i, o_2)$. Adding the two inequalities together and using again the triangle inequality, we have that
\begin{equation}
    \begin{aligned}
 MoA(\mathbf{o})&\ge\frac{1}{2n_{d'}}\cdot\sum_{i \in N_{d'} } [\delta(x_i, o_1)+\delta(x_i,o_2)]\ge\frac{1}{2}\cdot\delta(o_1,o_2).
\end{aligned}
\end{equation}

\noindent Therefore, we can obtain $MoA(\mathbf{w})\le 3\cdot MoA(\mathbf{o}).$

\noindent\textbf{Case 6:} \(o_1=w_1\le w_2<o_2 \). Similar to Case 5, we can obtain $MoA(\mathbf{w})\le 3\cdot MoA(\mathbf{o}).$

\noindent Above all, we can obtain an upper bound of $\max \left\{ 1 + \frac{2(1 - \alpha)}{\alpha}, \, 1 + \frac{2}{1 - \alpha},3 \right\}=\max \left\{ 1 + \frac{2(1 - \alpha)}{\alpha}, \, 1 + \frac{2}{1 - \alpha} \right\}$.
\qed

\paragraph{Proof of Theorem 7}

   Assume for  contradiction that there exists a strategyproof mechanism $M$ that has a  distortion strictly smaller than  $\frac{7}{2}-\varepsilon$, for some $\varepsilon>0$. Let $\theta>0$ be an infinitesimal.  We will reach a contradiction through considering the following several instances with the set of candidate locations $A=\left\{0,0,1,1,2,2\right\}$.

\textbf{Instance $I_1$}: Consider an instance $I_1$ with five agents in a single group, with the location profile $\mathbf{x^1}$, where $x_1^1=x_2^1=x_3^1=x_4^1=0$ and $x_5^1=1$. We can obtain that $MoA (0,0) = 1/5, MoA (0,1) = MoA (1,0)=1, MoA(1,1)=4/5, MoA (1,2) = MoA (2,1) = MoA (2,0) = MoA (0,2) = MoA (2,2) = 9/5.$  Clearly, $OPT(I_1)=(0,0)$, so it must be that  $M(I_1)=(0,0)$. Otherwise, the distortion of $M$ is at least $4$, which contradicts the assumption.

\textbf{Instance ($I_2$)}: Consider an instance $I_2$ with five agents in a single group, with the location profile $\mathbf{x^2}$,  where $x_1^2=\frac{1}{2}-\theta$ ,\ $x_2^2=x_3^2=x_4^2=0$ and $x_5^2=1$. If $M(I_2)\ne(0,0)$, then $c(x_1^2,M(I_2))\ge \frac{1}{2}+\theta$. Considering $c(x_1^2,M(I_1))=\frac{1}{2}-\theta$,  agent $1$ at $x_1^2$ can decrease her cost by misreporting her location as $0$, in contradiction to strategyproofness. Thus, $M(I_2)=(0,0)$.

\textbf{Instance $I_3$}: Consider an instance $I_3$ with five agents in a single group, with the location profile $\mathbf{x^3}$, where $x_1^3=x_2^3=x_3^3=x_4^3=\frac{1}{2}-\theta$ and $x_5^3=1$. According to Instance $I_2$, we can deduce that $M(I_3)=(0,0)$.

\textbf{Instance $I_4$}: Consider an instance $I_4$ with five agents in a single group, with the location profile $\mathbf{x^4}$, where $x_1^4=x_2^4=x_3^4=x_4^4=2$ and $x_5^4=1$.  Since Instance $I_4$ and Instance $I_1$ have symmetry with respect to $A$ , we have that $M(I_4)=(2,2)$.

\textbf{Instance $I_5$}: Consider an instance $I_5$ with five agents in a single group, with the location profile $\mathbf{x^5}$, where $x_1^5=x_2^5=x_3^5=x_4^5=\frac{3}{2}+\theta$ and $x_5^5=1$. Similar to Instance $I_3$,  we have that $M(I_5)=(2,2)$.

 Finally, to reach a contradiction, we consider the following Instance $I_6$ with two groups:

\begin{itemize}  
    \item In group $1$, there are four agents at \( \frac{1}{2} - \theta \) and one agent at $1$.  
    \item In group $2$, there are four agents at \( \frac{3}{2} + \theta \) and one agent at $1$.  
\end{itemize}  

\noindent Considering  $M(I_3)=(0,0)$ and $M(I_5)=(2,2)$,  by  \textit{(P1)} and \textit{(P2)}, the representatives of group $1$ and group $2$  are \( (0,0) \) and \( (2,2) \), respectively.
 Therefore, \( M(I_6) = (0,0), (0,2), (2,0) \) or \( (2,2) \) and we can obtain that  $MoA(M|I_6) = \frac{7 + 4\theta}{5}$.
However, \( OPT(I_6) = (1,1) \) and   $MoA(OPT|I_6) = \frac{2 + 4\theta}{5}$.  
That is, the distortion of mechanism \( M \) is at least   
$\frac{7 + 4\theta}{2 + 4\theta} \geq \frac{7}{2} - \varepsilon$,  for $\theta\le\frac{\epsilon}{5-2\epsilon}$; which is a contradiction.  
\qed

\paragraph{Proof of Lemma 2}

   Given any instance $I$, let \( \mathbf{w}= (w_1, w_2) \) be the solution chosen by the mechanism and \( \mathbf{o} = (o_1, o_2) \) be an optimal solution. W.l.o.g., we assume that \( w_1 \leq w_2 \) and \( o_1 \leq o_2 \).  According to the property of this mechanism, there exists a group $d^*$ such that \( w_1 = y_{d^*(1)} = t(r_{d^*}) \) and $w_2 = y_{d^*(2)} = s(r_{d^*})$. 
For each group $d$, denote by $i_d$ and $i'_{d}$ the agents in $N_d$ with the maximum  individual cost for $\mathbf{w}$ and $\mathbf{o}$, respectively. Then we have 
\begin{equation}
    \begin{aligned}
       &AoM(\mathbf{w}) =\frac{1}{k} \sum_{d \in D} \left\{\max_{i \in N_d} \delta(x_i, w(x_i))\right\} = \frac{1}{k}\sum_{d \in D} \delta(x_{i_{d}}, w(x_{i_{d}})),\\
      & AoM(\mathbf{o}) = \frac{1}{k}\sum_{d \in D}\left\{ \max_{i \in N_d} \delta(x_i, o(x_i))\right\} = \frac{1}{k}\sum_{d \in D} \delta(x_{i_{d'}}, o(x_{i'_{d}})).
\end{aligned}
\end{equation}
\noindent Let $D_1\subseteq D$ ($D_2\subseteq D$) be the set of groups where $w(x_{i_d})=w_1$ for any $d\in D_1$ ($w(x_{i_d})=w_2$ for any $d\in D_2$ ). 

\noindent\textbf{Case 1:} \( o_1 \leq o_2 < w_1 \leq w_2 \).  Using the triangle inequality, and since \( \delta(x_{i_d}, o(x_{i_{d}})) \leq \delta(x_{i'_{d}}, o(x_{i'_{d}})) \), we have  
\begin{equation}
    \begin{aligned}
       AoM(\mathbf{w})& = \frac{1}{k}\sum_{d \in D_1} \delta(x_{i_d}, w_1) +\frac{1}{k} \sum_{d \in D_2} \delta(x_{i_d}, w_2) \\
       &\leq \frac{1}{k}\sum_{d \in D} \delta(x_{i_d}, o(x_{i_d})) + \frac{|D_1|}{k} \cdot \delta(o_1, w_1) +\frac{|D_2|}{k}  \cdot \delta(o_2, w_2)\\
       &\leq AoM(\mathbf{o}) +  \delta(o_1, w_2).
\end{aligned}
\end{equation}

\noindent Let $S=\{ d\in D\mid y_{d(1)}\ge w_1, y_{d(2)}\ge w_1\}$.
 For each group $d\in D$, since $y_{d(1)}$, $y_{d(2)}$ are adjacent and,  \( w_1 \) is the \( (\beta \cdot k) \)-th leftmost location in \( \{ y_{d(1)} \}_{d \in D} \) and \( w_2 \) is the \( (\beta \cdot k) \)-th leftmost location in the  updated \( \{ z_{d} \}_{d \in D}\), we have that \( |S| \geq (1 - \beta) \cdot k \). Since \( w_1 \) and \( w_2 \) are the closest locations to \( x_{r_{d^*}} \), it holds that \( x_{r_{d^*}}\ge  \frac{o_2 + w_2}{2} \). Then, for any \( d \in S \), we have $\delta(x_{r_d}, o_1) \geq \delta(o_1, o_2) + \frac{\delta(o_2, w_2)}{2}$. Hence,
\begin{equation}
    \begin{aligned}
       AoM(\mathbf{o})\ge \frac{1}{k}\sum_{d \in S} \delta(x_{r_d}, o_1)  \ge \frac{(1-\beta)}{2}\cdot\delta(o_1,w_2).
\end{aligned}
\end{equation}

\noindent Now, by (25) and (26), we have that
\begin{equation}
    \begin{aligned}
       AoM(\mathbf{w})\leq AoM(\mathbf{o}) +  \delta(o_1, w_2) \le \left(1+\frac{2}{1-\beta}\right) \cdot AoM(\mathbf{o}).    
\end{aligned}
\end{equation}

\noindent\textbf{Case 2:} \( o_1 < o_2 = w_1 <w_2 \). Similar to Case 1, we can obtain $AoM(\mathbf{w})\le \left(1+\frac{2}{1-\beta}\right)\cdot AoM(\mathbf{o})$

\noindent\textbf{Case 3:} \( w_1 \le w_2 < o_1 \le o_2 \). Denote by \( L \) the set of \( (\beta \cdot k) \) groups from the one with the leftmost location in \( \{ y_{d(1)} \}_{d \in D} \) to the one with the \( (\beta \cdot k) \)-th location in \( \{ y_{d(1)} \}_{d \in D} \) and \( R \) the set of the remaining  groups. Then, for every \( d \in L \), $i\in N_d$, since \( y_{d(1)} \) is the closest location to \( x_{r_d} \) and \( y_{d(1)} \leq w_1 \le w_2< o_1 \leq o_2 \),  it holds that  $\delta(x_i, w(x_i)) = \delta(x_i,w_2)\le  \delta(x_i, o(x_i)) = \delta(x_i, o_2)$.
\noindent Now, by the triangle inequality, we have that 
\begin{equation}
    \begin{aligned}
      AoM(\mathbf{w})& =\frac{1}{k} \sum_{d \in L} \delta(x_{i_d}, w_2) +\frac{1}{k} \sum_{d \in R} \delta(x_{i_d}, w(x_{i_d})) \\
      &\le\frac{1}{k}\sum_{d \in L} \delta(x_{i_d}, o(x_{i_d})) +\frac{1}{k} \sum_{d \in R\cap D_1} \delta(x_{i_d}, w_1) +\frac{1}{k} \sum_{d \in R\cap D_2} \delta(x_{i_d}, w_2)  \\
      &\le AoM(\mathbf{o})+(1-\beta) \cdot \delta(o_2, w_1) .
\end{aligned}
\end{equation}

\noindent Since \( w_1 \) and \( w_2 \) are the closest  locations to \( x_{r_{d^*}} \), it holds that 

\[  
x_{r_{d^*}} \leq \frac{w_1 + w_2}{2} \quad \text{(when \( w_1 \neq w_2 \))} \quad \text{or} \quad x_{r_{d^*}} \leq \frac{w_2 + o_1}{2} \quad \text{(when \( w_1 = w_2 \))}  .
\]  Then, for every \( d \in L \), we have   $\delta(x_{r_d}, o_2) \geq  \frac{\delta(w_2, o_2) +\delta(w_1, w_2)}{2}.$

\noindent  Hence,  
\begin{equation}
    \begin{aligned}
     {AoM}(\mathbf{o}) \geq\frac{1}{k} \sum_{d \in L} \delta(x_{r_d}, o_2) \geq \frac{\beta }{2} \cdot \delta(o_2, w_1) . 
\end{aligned}
\end{equation}

\noindent Now, by (28) and(29) we have that 
\begin{equation}
    \begin{aligned}
      {AoM}(\mathbf{w})\le AoM(\mathbf{o})+(1-\beta)\cdot\delta(o_2,w_1)\le \left(1+\frac{2(1-\beta)}{\beta}\right)\cdot AoM(\mathbf{o}).
\end{aligned}
\end{equation}

\noindent\textbf{Case 4:} \( w_1 < w_2 = o_1 \le o_2 \).  Similar to Case 3, we can obtain ${AoM}(\mathbf{w})\le \left(1+\frac{2(1-\beta)}{\beta}\right)\cdot AoM(\mathbf{o})$.

\noindent\textbf{Case 5:} \( o_1 < w_1 \le w_2 \le o_2 \).  By the triangle inequality, we have that
\begin{equation}
    \begin{aligned}
      AoM(\mathbf{w})& =\frac{1}{k} \sum_{d \in D} \delta(x_{i_d}, w(x_{i_d}))\leq AoM(\mathbf{o}) + \delta(o_1, o_2).  
\end{aligned}
\end{equation}

\noindent Clearly, $AoM(\mathbf{o})\ge\frac{1}{k}\sum_{d \in D} \delta(x_{i'_{d}}, o_1)$ and $AoM(\mathbf{o})\ge\frac{1}{k}\sum_{d \in D} \delta(x_{i'_{d}}, o_2)$. Using again the triangle inequality, we have that
\begin{equation}
    \begin{aligned}
 AoM(\mathbf{o})&\ge\frac{1}{2k}\cdot\sum_{d \in D } [\delta(x_{i'_{d}}, o_1)+\delta(x_{i'_{d}},o_2)]\ge\frac{1}{2}\cdot\delta(o_1,o_2).
\end{aligned}
\end{equation}

\noindent Therefore, we can obtain $AoM(\mathbf{w})\le 3\cdot AoM(\mathbf{o}).$

\noindent\textbf{Case 6:} \(o_1=w_1\le w_2<o_2 \). Similar to Case 5, we can obtain $AoM(\mathbf{w})\le 3\cdot AoM(\mathbf{o}).$

\noindent Above all, we can obtain an upper bound of $\max \left\{ 1 + \frac{2(1 - \beta)}{\beta}, \, 1 + \frac{2}{1 - \beta},3 \right\}=\max \left\{ 1 + \frac{2(1 - \beta)}{\beta}, \, 1 + \frac{2}{1 - \beta} \right\}$.  
\qed

\paragraph{Proof of Theorem 9}

    Assume for a contradiction that there exists a strategyproof mechanism $M$ that has a distortion strictly smaller than $3-\epsilon$, for some $\epsilon>0$. Let $\theta>0$ be an infinitesimal. We will reach a contradiction through considering the following several instances with set of candidate locations $A=\left\{0,0,1,1\right\}$.

\textbf{Instance $I_1$}: Consider an instance $I_1$ with two agents in a single group, with the location profile $\mathbf{x^1}$, where $x_1^1=\frac{1}{2}$, $x_2^1=\frac{3}{2}$. We can obtain that $AoM (0,0) = AoM(0,1)=AoM(1,0)=\frac{3}{2}, AoM (1,1) = \frac{1}{2}$. Clearly, $OPT(I_1)=(1,1)$, so it must be that $M(I_1)=(1,1)$. Otherwise, the distortion of $M$ is at least $3$, which contradicts the assumption.

\textbf{Instance $I_2$}: Consider an instance $I_2$ with two agents in a single group, with the location profile $\mathbf{x^2}$, where $x_1^2=\frac{1}{2}$, $x_2^2=\frac{1}{2}+\theta$. If $M(I_2)\ne (1,1)$, then $c(x_2^2, M(I_2))=\frac{1}{2}+\theta$. Considering $c(x_2^2,M(I_1))=\frac{1}{2}-\theta$, agent 2 at $x_{2}^2$ can decrease her cost by misreporting her location as $\frac{3}{2}$, in contradiction to strategyproofness. Therefore, $M(I_2)=(1,1)$.

\textbf{Instance $I_3$}: Consider an instance $I_3$ with one agent in a single group, with the location profile $\mathbf{x^3}$, where $x_1^3=0$. In this case, the group representative chosen by mechanism $M$ must be $(0,0)$, i.e., $M(I_3)=(0,0)$. Otherwise, the distortion of $M$ is infinite.

\textbf{Instance $I_4$}: Consider an instance $I_4$ with one agent in a single group, with the location profile $\mathbf{x^4}$, where $x_{1}^4=\frac{1}{2}-\theta$. If $M(I_4)\ne (0,0)$, then $c(x_1^4, M(I_4))=\frac{1}{2}+\theta$. Considering $c(x_1^4,M(I_3))=\frac{1}{2}-\theta$, agent 1 at $x_{1}^4$ can decrease her cost by misreporting her location as $0$. Therefore, to maintain strategyproofness, $M(I_4)=(0,0)$.

\textbf{Instance $I_5$}: Consider an Instance $I_5$ with two groups:

\begin{itemize}  
    \item In group $1$, there is one agent at $\frac{1}{2}-\theta$, by Instance $I_4$, the representative  of this group is $(0,0)$.  
    \item In group $2$, there is one agent at $1$, similar to Instance $I_3$, the representative of this group is $(1,1)$. 
\end{itemize} 

\noindent   Since $AoM(0,0)=\frac{3}{4}-\frac{\theta}{2}$, $AoM(0,1)=AoM(1,0)=\frac{3}{4}+\frac{\theta}{2}$, $AoM(1,1)=\frac{1}{4}+\frac{\theta}{2}$, $M$ must output $(1,1)$ as the overall facility location profile, i.e., $M(I_5)=(1,1)$. Otherwise, the distortion is at least $\frac{\frac{3}{2}-\theta}{\frac{1}{2}+\theta}\ge3-\epsilon$, for $\theta\le\frac{\epsilon}{8-2\epsilon}$.
~\\

\noindent Now, we can reach a contradiction by considering the following instance $I_6$ with two groups:

\begin{itemize}  
    \item In  group $1$, there are two agents at $\frac{1}{2}$, $\frac{1}{2}+\theta$, respectively, by Instance $I_2$, the representative of this group is $(1,1)$.  
    \item In group $2$, there is one agent at $0$, by Instance $I_3$, the representative of this group is $(0,0)$. 
\end{itemize} 

\noindent According to Instance $I_5$ and \textit{(P3)}, we have that $M(I_6)=(1,1)$. Since $AoM(0,0)=\frac{1}{4}+\frac{\theta}{2},  AoM(0,1)=AoM(1,0)=\frac{3}{4}+\frac{\theta}{2},AoM(1,1)=\frac{3}{4}$, the distortion of $M$ is at least $\frac{\frac{3}{2}}{\frac{1}{2}+\theta}\ge3-\epsilon$, for $\theta\le\frac{\epsilon}{8-2\epsilon}$, which is a contradiction.
\qed

\end{appendix}

\end{document}